\documentclass[11pt]{article}
\usepackage{paralist}
\usepackage{color}
\usepackage{bm}

\usepackage{soul}

\usepackage[cm]{fullpage}
\usepackage[algo2e]{algorithm2e} 
\usepackage{amsfonts}
\usepackage{amssymb}

\usepackage{amsmath}
\usepackage{constants}
\usepackage{amsthm}
\makeatletter
\newtheorem*{rep@theorem}{\rep@title}
\newcommand{\newreptheorem}[2]{%
	\newenvironment{rep#1}[1]{%
		\def\rep@title{#2 \ref{##1}}%
		\begin{rep@theorem}}%
		{\end{rep@theorem}}}
\makeatother

\usepackage{ifpdf}
\newcommand{\mydriver}{hypertex}
\ifpdf
\renewcommand{\mydriver}{pdftex}
\fi
\usepackage[breaklinks,\mydriver]{hyperref}

\usepackage[margin=1in]{geometry}

\theoremstyle{plain}
\newtheorem{theorem}{Theorem}[section]%[chapter]
\newtheorem{fact}[theorem]{Fact}
\newtheorem{lemma}[theorem]{Lemma}
\newtheorem{corollary}[theorem]{Corollary}

\newtheorem{definition}[theorem]{Definition}

\usepackage{algorithm}
\usepackage{algpseudocode}
\usepackage[normalem]{ulem}

\algtext*{EndWhile}
\algtext*{EndIf}
\algtext*{EndFor}

\newcommand{\la}{\lambda}

\newcommand{\kk}{k}
\newenvironment{pfof}[1]{\begin{proof}[\textbf{Proof of #1: }]}{\end{proof}}

\theoremstyle{definition}

\title{Incremental Exact Min-Cut in Poly-logarithmic \\ Amortized Update Time}
\author{Gramoz Goranci\footnote{University of Vienna, Faculty of Computer Science, Vienna, Austria. E-mail: \texttt{gramoz.goranci@univie.ac.at}}
	\and
	Monika Henzinger\footnote{University of Vienna, Faculty of Computer Science, Vienna, Austria. E-mail: \texttt{monika.henzinger@univie.ac.at}}
	\and
	Mikkel Thorup\footnote{Faculty of Computer Science, University of Copenhagen, Denmark. E-mail: \texttt{mikkel2thorup@gmail.com}}
}
\date{}
\begin{document}
\begin{titlepage}
\maketitle

\thispagestyle{empty}
\begin{abstract}
We present a deterministic incremental algorithm for \textit{exactly} maintaining the size of a minimum cut with $\widetilde{O}(1)$ amortized time per edge insertion and $O(1)$ query time. This result partially answers an open question posed by Thorup [Combinatorica 2007]. It also stays in sharp contrast to a polynomial conditional lower-bound for the fully-dynamic weighted minimum cut problem. Our algorithm is obtained by combining a recent sparsification technique of Kawarabayashi and Thorup [STOC 2015] and an exact incremental algorithm of Henzinger [J. of Algorithm 1997].

We also study space-efficient incremental algorithms for the minimum cut problem. Concretely, we show that there exists an ${O}(n\log n/\varepsilon^2)$ space Monte-Carlo algorithm that can process a stream of edge insertions starting from an empty graph, and with high probability, the algorithm maintains a $(1+\varepsilon)$-approximation to the minimum cut. The algorithm has $\widetilde{O}(1)$ amortized update-time and constant query-time.
\end{abstract}

\end{titlepage}
\section{Introduction}
Computing a minimum cut of a graph is a fundamental algorithmic graph problem. While most of the focus has been on designing static efficient algorithms for finding a minimum cut, the dynamic maintenance of a minimum cut has also attracted increasing attention over the last two decades. The motivation for studying the dynamic setting is apparent, as real-life networks such as social or road network undergo constant and rapid changes.

Given an initial graph $G$, the goal of a dynamic graph algorithm is to build a data-structure that maintains $G$ and supports update and query operations. Depending on the types of update operations we allow, dynamic algorithms are classified into three main categories: (i) \emph{fully dynamic}, if update operations consist of both edge insertions and deletions, (ii) \emph{incremental}, if update operations consist of edge insertions only and (iii) \emph{decremental}, if update operations consist of edge deletions only. In this paper, we study incremental algorithms for maintaining the size of a minimum cut of an unweighted, undirected graph (denoted by $\lambda(G) = \lambda$) supporting the following operations:
\begin{itemize}
\item \textsc{Insert}$(u,v)$: Insert the edge $(u,v)$ in $G$.
\item \textsc{QuerySize}: Return the exact (approximate) size of a minimum cut of the current $G$.
\end{itemize}
For any $\alpha \geq 1$, we say that an algorithm is an $\alpha$-approximation of $\lambda$ if \textsc{QuerySize} returns a positive number $k$ such that $\lambda \leq k \leq \alpha \cdot \lambda$. Our problem is characterized by two time measures; \emph{query time}, which denotes the time needed to answer each query and \emph{total update time}, which denotes the time needed to process \emph{all} edge insertions. We say that an algorithm has an $O(t(n))$ amortized update time if it takes $O(m(t(n)))$ total update time for $m$ edge insertions starting from an empty graph. We use $\widetilde{O}(\cdot)$ to hide poly-logarithmic factors. 
\paragraph*{Related Work.} For over a decade, the best known static and deterministic algorithm for computing a minimum cut was due to Gabow~\cite{Gabow95} which runs in $O(m + \lambda^{2} \log n)$ time. Recently, Kawarabayashi and Thorup~\cite{thorup} devised a $\widetilde{O}(m)$ time algorithm which applies only to simple, unweighted and undirected graphs. Randomized Monte Carlo algorithms in the context of static minimum cut were initiated by Karger~\cite{Karger99}. The best known randomized algorithm is due to Karger~\cite{kargermin} and runs in $O(m \log^{3} n)$ time.  

Karger~\cite{Karger94} was the first to study the dynamic maintenance of a minimum cut in its full generality. He devised a fully dynamic, albeit randomized, algorithm for maintaining a  $\sqrt{1+2/\varepsilon}$-approximation of the minimum cut in $\widetilde{O}(n^{1/2 + \varepsilon})$ expected amortized time per edge operation. In the incremental setting, he showed that the update time for the same approximation ratio can be further improved to $\widetilde{O}(n^{\varepsilon})$. Thorup and Karger~\cite{thorupkarger} improved upon the above guarantees by achieving an approximation factor of $\sqrt{2+o(1)}$ and an $\widetilde{O}(1)$ expected amortized time per edge operation.

Henzinger~\cite{henzinger97} obtained the following guarantees for the incremental minimum cut; for any $\varepsilon \in (0,1]$, (i) an $O(1/\varepsilon^{2})$ amortized update-time for a $(2+\varepsilon)$-approximation, (ii) an $O(\log^{3} n / \varepsilon^{2})$ expected amortized update-time for a $(1+\varepsilon)$-approximation and (iii) an $O(\lambda \log n)$ amortized update-time for the exact minimum cut. 

For minimum cut up to some poly-logarithmic size, Thorup~\cite{fullythorup} gave a fully dynamic  Monte-Carlo algorithm for maintaining exact minimum cut in $\widetilde{O}(\sqrt{n})$ time per edge operation. He also showed how to obtain an $1+o(1)$-approximation of an arbitrary sized minimum cut with the same time bounds. In comparison to previous results, it is worth pointing out that his work achieves \textit{worst-case} update times.

Lacki and Sankowski~\cite{LackiS11} studied the dynamic maintenance of the exact size of the minimum cut in planar graphs with arbitrary edge weights. They obtained a fully dynamic algorithm with $\widetilde{O}(n^{5/6})$ worst-case query and update time.

There has been a growing interest in proving conditional lower bounds for dynamic problems in the last few years~\cite{abboud, henzinger15}. A recent result of Nanongkai and Saranurak~\cite{DS16} shows the following conditional lower-bound for the \emph{exact weighted} minimum cut assuming the Online Matrix-Vector Multiplication conjecture: for any $\varepsilon > 0$, there are no fully-dynamic algorithms with polynomial-time preprocessing that can simultaneously achieve $O(n^{1-\varepsilon})$ update-time and $O(n^{2-\varepsilon})$ query-time. 
\paragraph*{Our Results and Technical Overview.} We present two new incremental algorithms concerning the maintenance of the size of a minimum cut. Both algorithms apply to undirected, unweighted graphs. 
Our first and main result, presented in Section \ref{sec: exactMinCut}, shows that there is a deterministic incremental algorithm for \textit{exactly} maintaining the size of a minimum cut with $\widetilde{O}(1)$ amortized time per operation and $O(1)$ query time. This result allows us to  partially answer in the affirmative a question regarding efficient dynamic algorithms for exact minimum cut posed by Thorup~\cite{fullythorup}. Additionally, it also stays in sharp contrast to the polynomial conditional lower-bound for the fully-dynamic weighted minimum cut problem of ~\cite{DS16}. 

We obtain our result by heavily relying on a recent sparsification technique developed in the context of static minimum cut algorithms. Specifically, for a given simple graph $G$, Kawarabayashi and Thorup~\cite{thorup} designed an $\widetilde{O}(m)$ procedure that contracts vertex sets of $G$ and produces a multigraph $H$ with considerably fewer vertices and edges while preserving some family of cuts of size up to $(3/2)\lambda(G)$. Motivated by the properties of $H$, we crucially observe that it is ``safe'' to work entirely with graph $H$ as long as the sequence of newly inserted edges do not increase the size of the minimum cut in $H$ by more than $(3/2) \lambda(G)$. If the latter occurs, we recompute a new multigraph $H$ for the current graph $G$. Since $\lambda(G) \leq n$, the number of such re-computations is $O(\log n)$. For maintaining the minimum-cut of $H$, we appeal to the exact incremental algorithm due to Henzinger~\cite{henzinger97}. Though the combination of this two algorithms might seem immediate at first sight, it is not alone sufficient for achieving the claimed bounds. Our main contribution is to overcome some technical obstacles and formally argue that such combination indeed leads to our desirable guarantees.

Motivated by the recent work on \textit{space-efficient} dynamic algorithms~\cite{sayan, GibbKKT15}, we also study the efficient maintenance of the size of a minimum cut using only $\widetilde{O}(n)$ space.
% (see Section \ref{sec: ApproxMinCut}). 
Concretely, we present a ${O}(n\log n / \varepsilon^2)$ space Monte-Carlo algorithm that can process a stream of edge insertions starting from an empty graph, and with high probability, it maintains an $(1+\varepsilon)$-approximation to the minimum cut in ${O}(\alpha(n) \log^3 n /\varepsilon^2)$ amortized update-time and constant query-time. Note that none of the existing streaming algorithms for $(1+\varepsilon)$-approximate minimum cut~\cite{AhnG09,KelnerL13,AhnGM12} achieves these update and query times.

\section{Preliminary}

Let $G = (V,E)$ be an undirected, unweighted multi-graph with no self-loops. Two vertices $x$ and $y$ are $k$-\textit{edge connected} if there exist $k$ edge-disjoint paths connecting $x$ and $y$. A graph $G$ is $k$-\textit{edge connected} if every pair of vertices is $k$-edge connected. The \textit{local edge connectivity} $\lambda(G,x,y)$ of vertices $x$ and $y$ is the largest $k$ such that $x$ and $y$ are $k$-edge connected in $G$. The \textit{edge connectivity} $\lambda(G)$ of $G$ is the largest $k$ such that $G$ is $k$-edge connected.  

For a subset $S \subseteq V$ in $G$, the \textit{edge cut} $E_G(S, V \setminus S)$ is a set of edges that have one endpoint in $S$ and the other in $ V \setminus S$. We may omit the subscript when clear from the context. Let $\lambda(S,G) = |E_G(S, V \setminus S)|$ be the size of the edge cut. If $S$ is a singleton, we refer to such cut as a \textit{trivial} cut. Two vertices $x$ and $y$ are \textit{separated} by $E(S, V \setminus S)$ if they do not belong to the same connected component induced by the edge cut. A \textit{minimum edge cut} of $x$ and $y$ is a cut of minimum size among all cuts separating $x$ and $y$. A \textit{global minimum cut} $\lambda(G)$ for $G$ is the minimum edge cut over all pairs of vertices. By Menger's Theorem \cite{menger}, (a) the size of the minimum edge cut separating $x$ and $y$ is $\lambda(x,y,G)$, and (b) the size of the global minimum cut is equal to $\lambda(G)$.

Let $n$, $m_0$ and $m_1$ be the number of vertices, initial edges and inserted edges, respectively. The total number of edges $m$ is the sum of the initial and inserted edges. Moreover, let $\lambda$ and $\delta$ denote the size of the global minimum cut and the minimum degree in the final graph, respectively. Note that the minimum degree is always an upper bound on the edge connectivity, i.e., $\lambda \leq \delta$ and $m = m_0 + m_1 = \Omega{(\delta n)}$. 

A subset $U \subseteq V$ is \textit{contracted} if all vertices in $U$ are identified with some element of $U$ and all edges between them are discarded. For $G=(V,E)$ and a collection of vertex sets, let $H=(V_H,E_H)$ denote the graph obtained by contracting such vertex sets. Such contractions are associated with a mapping $h : V \rightarrow V_H$. For an edge subset $N \subseteq E$, let $N_h= \{(h(a),h(b)) : (a,b) \in N\} \subseteq E_H$ be its corresponding edge subset induced by $h$. 

%Throughout the paper we use $\widetilde{O}(\cdot)$ to hide poly-logarithmic factors. 

\section{Sparse certificates}  \label{sec: sparseCertificates}
In this section we review a useful sparsification tool, introduced by Nagamochi and Ibaraki~\cite{NagamochiI92}.
\begin{definition}[\cite{BenczurK15}] A \emph{sparse $k$-connectivity certificate}, or simply a \emph{$k$-certificate}, for an unweighted graph $G$ with $n$ vertices is a subgraph $G'$ of $G$ such that 
\begin{enumerate}
\item  $G'$ consists of at most $k(n-1)$ edges, and 
\item  $G'$ contains  all edges crossing cuts of size at most $k$.
\end{enumerate} \label{sparsedef}
\end{definition}

Given an undirected graph $G = (V,E)$, a \textit{maximal spanning forest decomposition (msfd)} $\mathcal{F}$ of order $k$ is a decomposition of $G$ into $k$ edge-disjoint spanning forests $F_i$, $1\leq i \leq k$, such that $F_i$ is a maximal spanning forest of $G \setminus (F_1 \cup F_2 \ldots \cup F_{i-1})$. Note that $G_k = (V, \bigcup_{i \leq k} F_i)$ is a $k$-certificate. An msfd fulfills the following property whose proof we defer to the appendix.

\begin{lemma}[\cite{NagamochiI}] \label{lemm: Nagamochi}
Let $\mathcal{F}=(F_1,\ldots,F_m)$ be an \emph{msfd} of order $m$ of a graph $G=(V,E)$, and let $k$ be an integer with $1 \leq k \leq m$. Then for any nonempty and proper subset $S \subset V$,
\[
	\lambda(S,G_k) \begin{cases}
    \geq k,& \text{if } \lambda(S,G) \geq k\\
    = \lambda(S,G)  & \text{if } \lambda(S,G) \leq k-1.
\end{cases}
\]
\end{lemma}
As $G_k$ is a subgraph of $G$, $\lambda(G_k) \leq \lambda(G)$. This implies that $\lambda(G_k) = \min(k,\lambda(G))$. 

Nagamochi and Ibaraki~\cite{NagamochiI92} presented an $O(m+n)$ time algorithm to construct a special msfd, which we refer to as DA-msfd.

%An msfd that fulfills the following additional properties is called a DA-msfd of order $k$: For a multigraph $G$, (1) for all $1 \leq i \leq k$, if $x$ and $y$ are connected in $F_i$, then they are $i$-edge connected in $G$; (2) $G$ is $k$-edge connected iff $G'$ is $k$-edge connected; (3) for any $1 \le i \le k$ and $x, y \in V$, $\lambda(\bigcup_{j \leq i} F_j, x, y) \geq \min(\lambda(G,x,y), i)$. As $G'$ is a subgraph of $G$, $\lambda(G') \le \lambda(G)$. This implies that $\lambda(G') = \min(k, \lambda(G))$.
 
\section{Incremental Exact Minimum Cut} \label{sec: exactMinCut}
In this section we present a deterministic incremental algorithm that exactly maintains $\lambda(G)$. The algorithm has an $\widetilde{O}(1)$ update-time, an $O(1)$ query time and it applies to any undirected, unweighted graph $G = (V,E)$. The result is obtained by carefully combining a recent result of Kawarabayashi and Thorup~\cite{thorup} on static min-cut and the incremental exact min-cut algorithm of Henzinger~\cite{henzinger97}. We start by describing the maintenance of non-trivial cuts, that is, cuts with at least two vertices on both sides.    
\paragraph*{Maintaining non-trivial cuts.} Kawarabayashi and Thorup~\cite{thorup} devised a near-linear time algorithm that contracts vertex sets of a simple input graph $G$ and produces a sparse multi-graph preserving all non-trivial minimum cuts of $G$. 
In the following theorem, we state a slightly generalized version of this algorithm.

\begin{theorem}[\textsc{KT-Sparsifier}~\cite{thorup}] Given an undirected, unweighted graph $G$ with $n$ vertices, $m$ edges, and min-cut $\lambda$, in $\widetilde{O}(m)$ time, we can contract vertex sets and produce a multigraph $H$ which consists of only $m_H = \widetilde{O}(m/\lambda)$ edges and $n_H = \widetilde{O}(n/\lambda)$ vertices, and which preserves all non-trivial minimum cuts along with the non-trivial cuts of size up to $(3/2) \lambda$ in $G$.
\label{SparsificationThm}
\end{theorem}

As far as non-trivial cuts are concerned, the above theorem implies that it is safe work on $H$ instead of $G$ as long as the sequence of newly inserted edges satisfies $\lambda_H \leq (3/2) \lambda$. To incrementally maintain the correct $\lambda_H$, we apply Henzinger's algorithm~\cite{henzinger97} on top of $H$. The basic idea to verify the correctness of the solution is to compute and store all min-cuts of $H$. Clearly, a solution is correct as long as an edge insertion does not increase the size of all min-cuts. If all min-cuts have increased, a new solution is computed using information about the previous solution. We next show how to do this efficiently. 

To store all minimum edge cuts we use the \textit{cactus tree} representation by Dinitz, Karzanov and Lomonosov~\cite{dinitz}. A cactus tree of a graph $G=(V,E)$ is a weighted graph $G_c = (V_c, E_c)$ defined as follows: There is a mapping $\phi: V \rightarrow V_c$ such that:
\begin{enumerate}
\item Every node in $V$ maps to exactly one node in $V_c$ and every node in $V_c$ corresponds to a (possibly empty) subset of $V$.
\item $\phi(x) = \phi(y)$ iff $x$ and $y$ are $(\lambda(G)+1)$-edge connected.
\item Every minimum cut in $G_c$ corresponds to a min-cut in $G$, and every min-cut in $G$ corresponds to \text{at least} one min-cut in $G_c$.
\item If $\lambda$ is odd, every edge of $E_c$ has weight $\lambda$ and $G_c$ is a tree. If $\lambda$ is even, $G_c$ consists of paths and simple cycles  sharing at most one vertex, where edges that belong to a cycle have weight $\lambda / 2$ while those not belonging to a cycle have weight $\lambda$.
\end{enumerate}
Dinitz and Westbrook~\cite{DinitzW98} showed that given a cactus tree, we can use the data structures from~\cite{GalilI93, Poutre00} to maintain the cactus tree for minimum cut size $\lambda$ under $u$ insertions, reporting when the minimum cut size increases to $\lambda+1$ in $O(u+n)$ total time. 

To quickly compute and update the cactus tree representation of a given multigraph $G$, we use an algorithm due to Gabow~\cite{Gabow91}. The algorithm computes first a subgraph of $G$, called a \textit{complete $\lambda$-intersection} or $I(G,\lambda)$, with at most $\lambda n$ edges, and uses $I(G,\lambda)$ to compute the cactus tree. Given some initial graph with $m_0$ edges, the algorithm computes $I(G,\lambda)$ and the cactus tree in $\widetilde{O}(m_0 + \lambda^{2}n)$ time. Moreover, given $I(G,\lambda)$ and a sequence of edge insertions that increase the minimum cut by 1, the new $I(G,\lambda)$ and the new cactus tree can be computed in $\widetilde{O}(m')$, where $m'$ is the number of edges in the current graph (this corresponds to one execution of the Round Robin subroutine~\cite{Gabow95}). 
\paragraph*{Maintaining trivial cuts.} We remark that the multigraph $H$ from Theorem \ref{SparsificationThm} preserves only non-trivial cuts of $G$. If $\lambda = \delta$, then we also need a way to keep track of a trivial minimum cut. We achieve this by maintaining a minimum heap $\mathcal{H}_G$ on the vertices, where each vertex is stored with its degree. If an edge insertion is performed, the values of the edge endpoints are updated accordingly in the heap. It is well known that constructing $\mathcal{H}_G$ takes $O(n)$ time. The supported operations \textsc{Min($\mathcal{H}_G$)} and \textsc{UpdateEndpoints($\mathcal{H}_G$,$e$)} can be implemented in $O(1)$ and $O(\log n)$ time, respectively (see \cite{Cormen}). 

This leads to Algorithm \ref{algo: ExactMinCut}.

\def\IF{\textbf{if}~}
\def\THEN{~\textbf{then}}
\def\ENDIF{\textbf{endif}~}
\def\WHILE{\textbf{while}~}
\def\ENDWHILE{\textbf{endwhile}~}
\def\ELSE{\textbf{else}~}
\def\GOTO{\textbf{Goto}~}
\def\SPACE{\quad~~}
\begin{algorithm}
\caption{\textsc{Incremental Exact Minimum Cut}}
\begin{algorithmic}[1]
\State Compute the size $\lambda_0$ of the min-cut of $G$ and set $\lambda^* = \lambda_0$.
\Statex Build a heap $\mathcal{H}_G$ on the vertices, where each vertex stores its degree as a key.
\Statex Compute a multigraph $H$ by running \textsc{KT-sparsifier} on $G$ and a mapping $h : V \rightarrow V_H$.
\Statex Compute the size $\lambda_H$ of the min-cut of $H$, a DA-msfd $F_1, \ldots, F_m$ of order $m$ of $H$, 
\Statex $I(H,\lambda_H)$, and a cactus-tree of $\bigcup_{i \leq \lambda_H+1} F_i$.
\State Set $N_h = \emptyset$. 
\Statex \WHILE there is at least one minimum cut of size \textsc{$\lambda_H$}~\textbf{do}
\Statex \SPACE \textbf{Receive the next operation}.
\Statex \SPACE  \IF it is a query~\textbf{then} \Return $\min$\{$\lambda_H$, \textsc{Min($\mathcal{H}_G$)}\} 
\Statex \SPACE \ELSE it is the insertion of an edge $(u,v)$, \textbf{then} 
\Statex \SPACE update the cactus tree according to the insertion of the new edge $(h(u),h(v))$,
\Statex \SPACE add the edge $(h(u),h(v))$ to $N_h$ and update the degrees of $u$ and $v$ in $\mathcal{H}_G$.
\Statex \SPACE \ENDIF
\Statex  \ENDWHILE
\Statex Set $\lambda_H = \lambda_H + 1$. %\texttt{// continued on next page}
%\algstore{testcont}
\State  \IF  $\min$\{$\lambda_H$, \textsc{Min($\mathcal{H}_G$)}\}$> (3/2) \lambda^{*}$\THEN 
\Statex \SPACE \texttt{// Full Rebuild Step}
\Statex \SPACE Compute $\lambda(G)$ and set $\lambda^{*} = {\lambda(G)}$. 
\Statex \SPACE Compute a multigraph $H$ by running \textsc{KT-sparsifier} on the current graph $G$.
\Statex \SPACE Update $\lambda_H$ to be the min-cut of $H$, compute a DA-msfd $F_1, \ldots, F_m$ of order $m$ of $H$,
\Statex \SPACE  and then $I(H, \lambda_H)$ and a cactus tree of $\bigcup_{i \leq \lambda_H+1} F_i$. 
\Statex \ELSE \IF $\lambda_H  \leq (3/2) \lambda^{*}$\THEN 
\Statex \SPACE  \SPACE // \texttt{Partial Rebuild Step}
\Statex  \SPACE \SPACE Compute a DA-msfd $F_1, \ldots, F_m$ of order $m$ of $\bigcup_{i \leq \lambda_H + 1} F_i \cup N_h$ and 
\Statex \SPACE \SPACE call the resulting forests $F_1,\ldots,F_m$.
\Statex \SPACE \SPACE Let $H' = (V_H,E')$ be a graph with $E' = I(H,\lambda_H - 1) \cup \bigcup_{i \leq \lambda_H + 1} F_i$.
\Statex \SPACE \SPACE Compute $I(H', \lambda_H)$ and a cactus tree of $H'$.
\Statex \SPACE \ELSE // \texttt{Special Step}
\Statex \SPACE \SPACE \WHILE \textsc{Min($\mathcal{H}_G$)} $\leq (3/2) \lambda^*$ ~\textbf{do}
\Statex \SPACE  \SPACE \SPACE \IF the next operation is a query~\textbf{then} \Return \textsc{Min($\mathcal{H}_G$)}
\Statex \SPACE \SPACE \SPACE \ELSE update the degrees of the edge endpoints in $\mathcal{H}_G$.
\Statex \SPACE \SPACE \SPACE \ENDIF
\Statex  \SPACE \SPACE \ENDWHILE
\Statex \SPACE\SPACE  \GOTO 3.
\Statex \SPACE \ENDIF
\Statex \ENDIF
\Statex \GOTO 2.
\end{algorithmic}
\label{algo: ExactMinCut}
\end{algorithm}

\paragraph*{Correctness.}
Let $G$ be the current graph throughout the execution of the algorithm and let $H$ be the corresponding multigraph maintained by the algorithm. Recall that $H$ preserves some family of cuts from $G$. We say that $H$ is \textit{correct} if and only if there exists a minimum cut from $G$ that is contained in the union of (a)  all trivial cuts of $G$ and (b) all cuts in $H$. Note that we consider $H$ to be correct even in the \texttt{Special Step} (i.e., when $\lambda_H > (3/2) \lambda^*$), where $H$ is not updated anymore since we are certain that the smallest trivial cut is smaller than any cut in $H$. 

To prove the correctness of the algorithm we will show that (1) it correctly maintains
a trivial min-cut at any time, (2) $H$ is correct as long as  $\min\{\textsc{Min}(\mathcal{H}_G),\lambda_H\} \leq (3/2) \lambda^{*}$
 (and when this condition fails we rebuild $H$), and (3) as long as $\lambda_H \leq (3/2) \lambda^*$, the algorithm correctly maintains  all cuts of size up to $\lambda_H + 1$ of $H$.

Let $N_h$ be the set of recently inserted edges in $H$ that the algorithm maintains during the execution of the \textbf{while} loop in Step 2.

\begin{lemma} \label{correctness1} Let $H=(V_H,E_H)$ be a multigraph with minimum cut $\lambda_H$ and let $N_h$ be a set with $N_h \subseteq E_H$. Further, let $F_1,\ldots,F_m$ be a DA-msfd of order $m \ge \lambda_H + 1$ of $H \setminus N_h$, and let $H'=(V_H, E')$ be a graph with $E' = N_h \cup \bigcup_{i \leq \lambda_H+1} F_i$. Then, a cut is a min-cut in $H'$ iff it is a min-cut in $H$. 
\end{lemma}
\begin{proof} 
We first show that every non-min cut in $H$ is a non-min cut in $H'$. By contrapositive, we get that a min-cut in $H'$ is a min-cut in $H$. 

To this end, let $(S,V_H \setminus S)$ be a cut with $|E_H(S,V_H \setminus S)| \geq \lambda_H +1$ in $H$.   Define $E_H(S, V_H \setminus S) \cap N_h = S_{N_h}$ and $E_H(S, V_H \setminus S) \cap (E_H \setminus N_h)  = S_{H \setminus N_h}$ such that $E_H(S,V_H \setminus S) = S_{N_h} \uplus S_{H \setminus N_h}$ and $|E_H(S,V_H \setminus S)| = |S_{N_h}| + |S_{H \setminus N_h}|$. Letting $F' = \bigcup_{i \leq \lambda_H+1} F_i$, we similarly define edge sets $S'_{N_h}$ and $S'_{F'}$ partitioning the edges $E'(S, V_H \setminus S)$ that cross the cut $(S, V_H \setminus S)$ in $H'$. First, observe that $S_{N_h} = S'_{N_h}$ since edges of $N_h$ are always included in $H'$. In addition, by Lemma \ref{lemm: Nagamochi}, we know that $F'$ preserves all cuts of $H \setminus N_h$ up to size $\lambda_H+1$. Thus, if $|S_{H \setminus N_h}| \leq \lambda_H + 1$ (Case 1), we get that $|S_{H \setminus N_h}| = |S'_{F'}|$. It follows that $|E'(S,V_H \setminus S)| = |S'_{N_h}| + |S'_{F'}| = |S_{N_h}| + |S_{H \setminus N_h}| = |E_H(S,V_H \setminus S)| \geq \lambda_H +1$. If $|S_{H \setminus N_h}| > \lambda_H + 1$ (Case 2), then $F'$ must contain at least $\lambda_H + 1$ edges crossing such cut and thus $|S'_{F'}| \geq \lambda_H + 1$. The latter implies that $|E'(S,V_H \setminus S)| = |S'_{N_h}| + |S'_{F'}| \geq \lambda_H + 1$. In both cases, $H'$ being a subgraph of $H$ implies that $\lambda(H') \leq \lambda_H$. Thus $(S, V_H \setminus S)$ cannot be a min-cut in $H'$. 

For the other direction, %let $(S, V_H \setminus S)$ be a min-cut in $H$. Since $H'$ is a subgraph of $H$, we know that $|E'(S, V_H \setminus S)| \leq \lambda_H$. Therefore, showing that $|E'(S, V_H \setminus S)| \geq \lambda_H$ implies that $(S, V_H \setminus S)$ is also a min cut in $H'$. 
consider a min-cut $(D,V_H \setminus D)$ of size $|E'(D, V_H \setminus D)|$ in $H'$. Let $D_{N_{h}}, D_{H \setminus N_h}, D'_{F'}, D'_{N_h}$ be defined as above. Considering the cut $(D, V_H \setminus D)$ in $H$, we know that $|E_H(D, V_H \setminus D)|= |D_{N_h}| + |D_{H \setminus N_h}| \geq \lambda_H$. We first note that $D_{N_h} = D'_{N_h}$ since edges of $N_h$ are always included in $H'$. Then, similarly as above, by Lemma \ref{lemm: Nagamochi}  we know that if $|D_{H \setminus N_h}| \leq \lambda_H + 1$, then $|E'(D, V_H \setminus D)| = |D'_{N_h}| + |D'_{F'}| = |D_{N_h}| + |D_{H \setminus N_h}| = |E_H(D, V_H \setminus D)| \geq \lambda_H$. If $|D_{H \setminus N_H}| > \lambda_H + 1$, then $F'$ must contain at least $\lambda_H + 1$ edges crossing such cut and thus $|E'(D, V_H \setminus D)| \geq \lambda_H + 1$. Combining both bounds we obtain that  $|E'(D, V_H \setminus D)| \geq \lambda_H$. Since $(D,V_H \setminus D)$ was chosen arbitrarily, we get that $\lambda(H') \geq \lambda_H$ must hold. The latter along with $\lambda(H') \leq \lambda_H$ imply that $\lambda(H') = \lambda_H$. 
%in particular, $|E'(S, V_H \setminus S)| \geq \lambda_H$. 
\end{proof}

\begin{lemma} \label{correctness3} The algorithm correctly maintains a trivial min-cut in $G$.
\end{lemma}
\begin{proof}
This follows directly from the min-heap property of $\mathcal{H}_G$.
\end{proof}

To simplify our notation, in the following we will refer to Step 1 as a \texttt{Full Rebuild Step} (namely the initial \texttt{Full Rebuild Step}).
\begin{lemma} \label{correctness2} For some current graph G, let $H$ be the maintained multi-graph of $G$ under the vertex mapping $h$ and assume that $\lambda_H \leq (3/2) \lambda^{*}$, where $\lambda^{*}$ denotes the min-cut of $G$ at the last \texttt{\em{Full Rebuild Step}}. Then the algorithm correctly maintains $\lambda_H = \lambda(H)$.
\end{lemma}
\begin{proof}
At the time of the last \texttt{Full Rebuild Step}, the algorithm calls \textsc{KT-sparsifier} on $G$, which yields a multigraph $H$ that preserves all non-trivial min-cuts of $G$. The value of $\lambda_H$ is updated to $\lambda(H)$ and a DA-msfd and a cactus tree are constructed for $H$. The latter preserve all cuts of $H$ of size up to $\lambda_H +1$. Thus, the value of $\lambda_H$ is correct at this step.

Now suppose that the graph after the last \texttt{Full Rebuild Step} has undergone a sequence of edge insertions, which resulted in the current graph $G$ and its corresponding multigraph $H$ under the vertex mapping $h$. During these insertions, as long as $\lambda_H \leq (3/2)\lambda^*$, a sequence of $k$ \texttt{Partial Rebuild Steps} is executed, for some $k\geq 1$. Let $\lambda_H^{(i)}$ be the value of $\lambda_H$ after the $i$-th execution of \texttt{Partial Rebuild Step}, where $1\leq i \leq k$. Since, $\lambda_H^{(k)} = \lambda(H)$, it suffices to show that $\lambda_H^{(k)}$ is correct. We proceed by induction.

For the base case, we show that $\lambda_H^{(1)}$ is correct. First, using the fact that $\lambda_H$ and the cactus tree are correct at the last \texttt{Full Rebuild Step} and that the incremental cactus tree algorithm correctly tell us when to increment $\lambda_H$, we conclude that incrementing the value of $\lambda_H$ in Step 2 is valid. Thus, $\lambda_H^{(1)}$ is correct. Next, in a \texttt{Partial Rebuild Step}, the algorithm sparsifies the graph while preserving all cuts of size up to $\lambda_H^{(1)} + 1$ and producing a new cactus tree for the next insertions. The correctness of the sparsification follows from Lemma \ref{correctness1}.

For the induction step, let us assume that $\lambda_H^{(k-1)}$ is correct. Then, similarly to the base case, the correctness of $\lambda_H^{(k-1)}$, the cactus tree from the $(k-1)$-st \texttt{Partial Rebuild Step} and the correctness of the incremental cactus tree algorithm give that incrementing the value of $\lambda_H^{(k-1)}$ in Step 2 is valid and yields a correct $\lambda_H^{(k)}$.
\end{proof}
Note that when $\lambda_H > (3/2) \la^*$, the above lemma is not guaranteed to hold as the algorithm does not execute a \texttt{Partial Rebuild Step} in this case. However, we will show below that this is not necessary for the correctness of the algorithm. The fact that we do not need to execute a \texttt{Partial Rebuild Step} in this setting is crucial for achieving our time bound.

\begin{lemma} \label{correctness4} If $\min\{\textsc{Min}(\mathcal{H}_G),\lambda_H\} \leq 3/2 \lambda^{*}$, then $H$ is correct. 
\end{lemma}
\begin{proof}
Let $(S',V \setminus S')$ be any non-trivial cut in $G$ that is not in $H$. Such a cut must have cardinality strictly greater than $(3/2) \lambda^{*}$ since otherwise it would be contained in $H$. We show that $(S',V \setminus S')$ cannot be a minimum cut as long as $\min\{\textsc{Min}(\mathcal{H}_G),\lambda_H\} \leq (3/2) \lambda^{*}$ holds. We distinguish two cases.
\begin{enumerate}
\item If $\lambda_H \leq (3/2) \lambda^*$, then by Lemma \ref{correctness2} the algorithm maintains $\lambda_H$ correctly. Since $H$ is obtained from $G$ by contracting vertex sets, there is a cut $(S,V_H,S)$ in $H$, and thus in $G$, of value $\lambda_H$. It follows that $(S',V \setminus S')$ cannot be a minimum cut of $G$ since $|E(S', V \setminus S')| > (3/2) \lambda^* \geq \lambda_H = \lambda(H) \geq \lambda(G)$, where the last inequality follows from the fact that $H$ is a contraction of $G$.

\item If $\textsc{Min}(\mathcal{H}_G) \leq (3/2) \lambda^{*}$, then by Lemma \ref{correctness3} there is a cut of size $\textsc{Min}(\mathcal{H}_G) = \delta$ in $G$. Similarly, $(S', V \setminus S')$ cannot be a minimum cut of $G$ since $|E(S', V \setminus S')| > (3/2) \lambda^{*} \geq \delta \geq \lambda(G)$.
\end{enumerate}
Appealing to the above cases, we conclude $H$ is correct since a min-cut of $G$ is either contained in $H$ or it is a trivial cut of $G$. 
\end{proof}

\begin{lemma} Let $G$ be some current graph. Then the algorithm correctly maintains $\lambda(G)$.
\end{lemma}
\begin{proof} 
Let $G$ be some current graph and $H$ be the maintained multi-graph of $G$ under the vertex mapping $h$. %By definition $\lambda(G) = \min_{(S, V \setminus S)} |E(S,V \setminus S)|$. 
We will argue that $\lambda(G) = \min\{\textsc{Min}(\mathcal{H}_G),\lambda_H\}$. 

If $\min\{\textsc{Min}(\mathcal{H}_G),\lambda_H\} \leq (3/2) \lambda^{*}$, then by Lemma \ref{correctness4}, $H$ is correct i.e., there exists a minimum cut of $G$ that is contained in the union of all trivial cuts of $G$ and all cuts in $H$. Lemma \ref{correctness3} guarantees that the algorithm correctly maintains $\textsc{Min}(\mathcal{H}_G)$, i.e., the trivial minimum cut of $G$. If $\lambda_H \leq (3/2) \lambda^*$, then Lemma \ref{correctness2} ensures that $\lambda_H = \lambda(H)$, and thus $\min\{\textsc{Min}(\mathcal{H}_G), \lambda_H\} = \lambda(G)$.  If, however, $\lambda_H > (3/2) \la^*$ but $\min\{\textsc{Min}(\mathcal{H}_G),\lambda_H\} \leq (3/2) \lambda^{*}$, then $\lambda_H  > \min\{\textsc{Min}(\mathcal{H}_G),\lambda_H\}$ which implies that $\min\{\textsc{Min}(\mathcal{H}_G), \lambda_H\} = \textsc{Min} (\mathcal{H}_G) = \lambda(G)$. As we argued above, the algorithm correctly maintains $\textsc{Min}(\mathcal{H}_G)$ at any time. Thus it follows that the algorithm correctly maintains $\lambda(G)$ in this case as well.

The only case that remains to consider is $\textsc{Min}(\mathcal{H}_G) > (3/2)\lambda^{*}$ and $\lambda_H > (3/2)\lambda^{*}$. But this implies that $\min\{\textsc{Min}(\mathcal{H}_G),\lambda_H\} > (3/2) \lambda^{*}$, and the algorithm computes a $H$ and $\lambda(G)$ from scratch and sets $\lambda_H$ correctly. After this full rebuild $\la(G) = \min\{\textsc{Min}(\mathcal{H}_G),\lambda_H\}$ trivially holds.
\end{proof}
 
\paragraph*{Running Time Analysis.}

\begin{theorem} Let $G$ be a simple graph with $n$ nodes and $m_0$ edges. Then the total time for inserting $m_1$ edges and maintaining a minimum edge cut of $G$ is $ \widetilde{O}(m_0 + m_1)$. If we start with an empty graph, the amortized time per edge insertion is $\widetilde{O}(1)$. The size of a minimum cut can be answered in constant time.
\end{theorem}
\begin{proof}
We first analyse Step 1. Building the heap $\mathcal{H}_G$ and computing $\lambda_0$ take $O(n)$ and $\widetilde{O}(m_0)$ time, respectively. The total running time for constructing $H$, $I(H,\lambda_H)$ and the cactus tree is dominated by $\widetilde{O}(m_0 + \lambda^{2}_0 \cdot ( n / \lambda_0)) = \widetilde{O}(m_0)$. Thus, the total time for Step 1 is $\widetilde{O}(m_0)$.

Let $\lambda_H^0, \ldots, \lambda_H^f$ be the values that $\lambda_H$ assumes in Step 2  during the execution of the algorithm in increasing order. We define \text{Phase} $i$ to be all steps executed after Step $1$ while $\lambda_H = \lambda_H^{i}$, excluding Full Rebuild Steps and Special Steps. Additionally, let $\lambda^{*}_0, \ldots, \lambda^{*}_{O(\log n)}$ be the values that $\lambda^{*}$ assumes during the algorithm. 
We define \textit{Superphase} $j$ to consist of the $j$-th \texttt{Full Rebuild Step} along with all steps executed while $\min\{\textsc{Min}(\mathcal{H}_G), \lambda_H\} \leq (3/2) \lambda^{*}_j$, where $\lambda^{*}_j$ is the value of $\lambda(G)$ at the \texttt{Full Rebuild Step}. Note that a superphase consists of a sequence of phases and potentially a final \texttt{Special Step}. Moreover, the algorithm runs a phase if $\lambda_H \leq (3/2) \lambda^{*}$. 

We say that $\lambda_H^i$ \textit{belongs} to superphase $j$, if the $i$-th phase is executed during superphase $j$ and $\lambda_H^i\leq (3/2) \lambda_j^{*}$. We remark that the number of vertices in $H$ changes only at the beginning of a superphase, and remains unchanged during its lifespan.

Let $n_j$ denote the number of vertices in some superphase $j$. We bound this quantity as follows:
\begin{fact} \label{fact}
Let $j$ be a superphase during the execution of the algorithm. Then, we have
\[
	n_j = \widetilde{O}(n / \lambda_H^i), \text{ for all } \lambda_H^i \text{ belonging to superphase } j.
\]
\end{fact}
\begin{proof}
From Step 3 we know that $n_j = \widetilde{O}(n / \lambda^{*}_j)$. Moreover, observe that $\lambda_j^{*} \leq \lambda_H^i$ and a phase is executed whenever $\lambda_H^i \leq (3/2) \lambda_j^{*}$. Thus, for all $\lambda_H^i$'s belonging to superphase $j$, we get the following relation
\begin{equation}
\label{MinCutRelation}
	\lambda^{*}_j \leq \lambda_H^i \leq (3/2) \lambda^{*}_j,
\end{equation}
which in turn implies that $n_j = \widetilde{O}(n / \lambda^{*}_j) = \widetilde{O}(n / \lambda_H^i)$.
\end{proof}

For the remaining steps, we divide the running time analysis into two parts (one part corresponding to phases, and the other to superphases). 

\paragraph*{Part $1$.}For some superphase $j$, the $i$-th phase consists of the $i$-th execution of a \texttt{Partial Rebuild Step} followed by the execution of Step 2. Let $u_i$ be the number of edge insertions in Phase $i$. The total time for Step 2 is $O(n_j+u_i \log n) = \widetilde{O}(n + u_i)$. Using Fact $9$, we observe that $\bigcup_{i \leq \lambda_H +1}F_i \cup N_h$ has size $O(u_{i-1} + \lambda^{i}_H n_j) = \widetilde{O}(u_{i-1} + n)$. Thus, the total time for computing DA-msfd in a \texttt{Partial Rebuild Step} is $\widetilde{O}(u_{i-1} + n)$. Similarly, since $H'$ has $O(\lambda_H^{i} n_j) = \widetilde{O}(n)$ edges, it takes $\widetilde{O}(n)$ time to compute $I(H',\lambda_{H}^i)$ and the new cactus tree.

The total time spent in Phase $i$ is $\widetilde{O}(u_{i-1} + u_{i} + n)$. Let $\lambda$ and $\lambda_H$ denote the size of the minimum cut in the final graph and its corresponding multigraph, respectively. Note that $\sum_{i=1}^{\lambda} u_i \leq m_1$, $\lambda n \leq m_0 + m_1$ and recall Eqn. (\ref{MinCutRelation}).  This gives that the total work over all phases is
\[
    \sum_{i = 1}^{\lambda_H} \widetilde{O}\left(u_{i-1} + u_{i} + n\right) = \sum_{i = 1}^{\lambda} \widetilde{O}\left(u_{i-1} + u_{i} + n\right) = \widetilde{O}(m_0 + m_1).
\]

\paragraph*{Part $2$.}The $j$-th superphase consists of the $j$-th execution of a \texttt{Full Rebuild Step} along with a possible execution of a \texttt{Special Step}, depending on whether the condition is met. In a \texttt{Full Rebuild Step}, the total running time for constructing $H$, $I(H,\lambda^{*}_j)$ and the cactus tree is dominated by $\widetilde{O}(m_0 + m_1 + (\lambda^{*}_j)^{2} \cdot (n / \lambda^{*}_j)) = \widetilde{O}(m_0 + m_1)$. The running time of a Special Step is $\widetilde{O}(m_1)$.

Throughout its execution, the algorithm begins a new superphase whenever $\lambda(G) =\min$ $\{\textsc{Min}(\mathcal{H}_G), \lambda_H\} > (3/2)\lambda^{*}$. This implies that $\lambda(G)$ must be at least $(3/2)\lambda^{*}$, where $\lambda^{*}$ is the value of $\lambda(G)$ at the last \texttt{Full Rebuild Step}. Thus, a new superphase begins whenever $\lambda(G)$ has increased by a factor of $3/2$, i.e., only $O(\log n)$ times over all insertions. 
This gives that the total time over all superphases is $\widetilde{O}(m_0 + m_1)$. 
\end{proof}

\section{Incremental \texorpdfstring{$(1+\varepsilon)$}{1+eps} Minimum Cut with \texorpdfstring{$\widetilde{O}(n)$}{O(n poly log n)} space} \label{sec: ApproxMinCut}
In this section we present two $\widetilde{O}(n)$ space incremental Monte-Carlo algorithms that w.h.p  maintain the size of a min-cut up to a $(1+\varepsilon)$-factor. Both algorithms have $\widetilde{O}(1)$ update-time and $\widetilde{O}(1)$, resp.~$O(1)$ query-time. The first algorithm uses $O(n \log^{2}n / \varepsilon^2)$ space, while the second one improves the space complexity to $O(n \log n / \varepsilon^2)$. 

\subsection{An \texorpdfstring{$O(n \log^2 n / \varepsilon^2)$}{O (n log2 n)} space algorithm}
Our first algorithm follows an approach that was used
in several previous work~\cite{henzinger97, thorupkarger, fullythorup}. The basic idea is to maintain the min-cut up to some size $k$ using small space. We achieve this by maintaining a sparse $(k+1)$-certificate and incorporating it into the incremental exact min-cut algorithm due to Henzinger~\cite{henzinger97}, as described in Section \ref{sec: exactMinCut}. Finally we apply the well-known randomized sparsification result due to Karger~\cite{Karger99} to obtain our result.

\paragraph*{Maintaining min-cut up to size $k$ using $O(kn)$ space.} We incrementally maintain an msfd for an unweighted graph $G$ using $k+1$ union-find data structures $\mathcal{F}_1, \ldots, \mathcal{F}_{k+1}$ (see~\cite{Cormen}). Each $\mathcal{F}_i$ maintains a spanning forest $F_i$ of $G$. Recall that $F_1,\ldots,F_{k+1}$ are edge-disjoint. 
When a new edge $e=(u,v)$ is inserted into $G$, we define $i$ to be the first index such that $\mathcal{F}_i.$\textsc{Find}$(u)$ $\neq$ $\mathcal{F}_i.$\textsc{Find}$(v)$. If we found such an $i$, we append the edge $e$ to the forest $F_i$ by setting $\mathcal{F}_i.$\textsc{Union}$(u,v)$ and return $i$. If such an $i$ cannot be found after $k+1$ steps, we simply discard edge $e$ and return NULL. We refer to such procedure as $(k+1)$-\textsc{Connectivity}$(e)$.

It is easy to see that the forests maintained by $(k+1)$-\textsc{Connectivity}$(e)$ for every newly inserted edge $e$ are indeed edge-disjoint. Combining this procedure with techniques from Henzinger~\cite{henzinger97} leads to the following Algorithm \ref{algo: ExactMinCutUpToK}. 

\begin{algorithm}
\caption{\textsc{Incremental Exact Min-Cut up to size $k$}}
\begin{algorithmic}[1]
\State Set $\lambda = 0$, initialize $k$ union-find data structures $\mathcal{F}_1, \ldots, \mathcal{F}_{k+1}$, 
\Statex $k$ empty forests $F_1,\ldots,F_{k+1}$, $I(G,\lambda)$, and an empty cactus tree. 
\State \WHILE there is at least one minimum cut of size \textsc{$\lambda$}~\textbf{do}
\Statex \SPACE \textbf{Receive the next operation}.
\Statex \SPACE  \IF it is a query~\textbf{then} \Return $\lambda$ 
\Statex \SPACE \ELSE it is the insertion of an edge $e$,~\textbf{then}
\Statex \SPACE Set $i =$ $(k+1)$-\textsc{Connectivity}$(e)$.
\Statex \SPACE \SPACE~ \IF $i \neq $ NULL ~\textbf{then}
\Statex \SPACE \SPACE \SPACE Set $F_i = F_i \cup \{e\}$. 
\Statex \SPACE \SPACE \SPACE Update the cactus tree according to the insertion of the edge $e$.
\Statex \SPACE \SPACE ~~\ENDIF
\Statex \SPACE \ENDIF
\Statex  \ENDWHILE
\State  Set $\lambda = \lambda + 1$.
\Statex  Let $G' = (V,E')$ be a graph with $E' = I(G,\lambda - 1) \cup \bigcup_{i \leq \lambda+1} F_i$.
\Statex  Compute $I(G',\lambda)$ and a cactus tree of $G'$.
\Statex \GOTO 2.
%\algstore{testcont}
\end{algorithmic}
\label{algo: ExactMinCutUpToK}
\end{algorithm}
The correctness of the above algorithm is immediate from Lemmas \ref{correctness1} and \ref{correctness2}. The running time and query bounds follow from Theorem 8 of Henzinger~\cite{henzinger97}. For the sake of completeness, we provide here a full proof.

\begin{corollary} \label{cor: ExactPolyLog}
For $k > 0$, there is an $O(kn)$ space algorithm that processes a stream of edge insertions starting from any empty graph $G$ and maintains an exact value of $\min\{\lambda(G),k\}$. Starting from an empty graph, the total time for inserting $m$ edges is $O(km\alpha(n) \log n )$ and queries can be answered in constant time, where $\alpha(n)$ stands for the inverse of Ackermann function.
\end{corollary}
\begin{proof}
We first analyse Step $1$. Initializing $k+1$ union-find data structures takes $O(kn)$ time. The running time for constructing $I(G,\lambda)$ and building an empty cactus tree is also dominated by $O(kn)$. Thus, the total time for Step $1$ is $O(kn)$.

Let $\lambda_0, \ldots, \lambda_f$, where $\lambda_f \leq k$,  be the values that $\lambda$ assumes in Step $2$ during the execution of the algorithm in increasing order. We define Phase $i$ to be all steps executed while $\lambda = \lambda_i$. For $i\geq 1$, we can view Phase $i$ as the $i$-th execution of Step $3$ followed by the execution of Step $2$. Let $u_i$ denote the number of edge insertion in Phase $i$. The total time for testing the $(k+1)$-connectivity of the endpoints of the newly inserted edges, and updating the cactus tree in Step $2$ is dominated by $O(n + k \alpha(n) u_i)$. Since the graph $G'$ in Step $3$ has always at most $O(kn)$ edges, the running time to compute $I(G',\lambda)$ and the cactus tree of $G'$ is $O(kn \log n)$. Combining the above bounds, the total time spent in Phase $i$ is $O(k(\alpha(n)u_i + n \log n))$. Thus, the total work over all phases is $O(km\alpha(n) \log n)$. 

The space complexity of the algorithm is only $O(kn)$, since we always maintain at most $k+1$ spanning forests during its execution.
\end{proof}

\paragraph*{Dealing with min-cuts of arbitrary size.} We observe that Corollary \ref{cor: ExactPolyLog} gives polylogarithmic amortized update time only for min-cuts up to some polylogarithmic size. For dealing with min-cuts of arbitrary size, we use the well-known sampling technique due to Karger~\cite{Karger99}. This allows us to get an $(1+\varepsilon)$-approximation to the value of a min-cut with high probability.

\begin{lemma}[\cite{Karger99}] \label{lemm: Karger} Let $G$ be any graph with minimum cut $\lambda$ and let $p \geq 12(\log n)/(\varepsilon^{2}\lambda)$. Let $G(p)$ be a subgraph of $G$ obtained by including each of edge of $G$ to $G(p)$ with probability $p$ independently. Then the probability that the value of any cut of $G(p)$ has value more than $(1+\varepsilon)$ or less than $(1-\varepsilon)$ times its expected value is $O(1/n^{4})$.  
\end{lemma} 

For some integer $i \geq 1$, let $G_i$ denote a subgraph of $G$ obtained by including each edge of $G$ to $G_i$ with probability $1/2^{i}$ independently. We now have all necessary tools to present our incremental algorithm.

\begin{algorithm} 
\caption{\textsc{$(1+\varepsilon)$-Min-Cut with $O(n \log^2 n / \varepsilon^{2})$ Space}}
\begin{algorithmic}[1]
\State \textbf{For} $i=0,\ldots, \lfloor \log n \rfloor$, let $G_i$ be an initially empty sampled subgraph. 
\State \textbf{Receive the next operation}.
\Statex  \IF it is a query~\textbf{then}
\Statex \SPACE Find the minimum $j$ such that $\lambda(G_j) \leq k$ and  \Return $2^{j}\lambda(G_j)/(1-\varepsilon)$.
\Statex \ELSE it is the insertion of an edge $e$,~\textbf{then}
\Statex \SPACE Include edge $e$ to each $G_i$ with probability $1/2^{i}$.
\Statex \SPACE Maintain the exact min cut of each $G_i$ up to size $k=48 \log n / \varepsilon^2$ using Algorithm \ref{algo: ExactMinCutUpToK}.
\Statex  \ENDIF
\State \GOTO 2.

%\algstore{testcont}
\end{algorithmic}
\label{algo: SmallSpaceMinCut1}
\end{algorithm}

%% Old Interface of the Algorithm
%\begin{enumerate}
%\item For $i=0,\ldots, \lfloor \log n \rfloor$, let $G_i$ be the initially empty sampled subgraphs.
%\item If an edge $e$ is inserted into $G$, include $e$ to each $G_i$ with probability $1/2^{i}$ and maintain the exact minimum cut of $G_i$ up to size $k = 40 \log n / \varepsilon^{2}$ using Algorithm \ref{algo: ExactMinCutUpToK}.
%\item If the operation is a query, find the minimum $j$ such that the min-cut of $G_j$ is at most $k$. Return  $2^{j}\lambda(G_j)/(1-\varepsilon)$.  
%\end{enumerate}

\begin{theorem} \label{thm: space1}
There is an $O(n \log^{2} n/\varepsilon^{2})$ space randomized algorithm that processes a stream of edge insertions starting from an empty graph $G$ and maintains a $(1+\varepsilon)$-approximation to a min-cut of $G$ with high probability. The amortized update time per operation is $O(\alpha(n)\log^{3} n / \varepsilon^{2})$ and queries can be answered in $O(\log n)$ time. 
\end{theorem}
\begin{proof}
We first prove the correctness of the algorithm. For an integer $t \geq 0$, let $G^{(t)} = (V,E^{(t)})$ be the graph after the first $t$ edge insertions. Further, let $\lambda(G^{(t)})$ denote the min-cut of $G^{(t)}$, $p^{(t)}=12(\log n)/(\varepsilon^{2}\lambda^{(t)})$ and $\lambda(G,S) = |E_G(S, V \setminus S)|$, for some cut $(S, V \setminus S)$. For any integer $i \leq \lfloor \log_2 1 / p^{(t)} \rfloor$, Lemma \ref{lemm: Karger} implies that for any cut $(S,V \setminus S)$, $(1-\varepsilon)/2^{i} \lambda(G^{(t)},S) \leq \lambda(G_{i}^{(t)},S) \leq (1+\varepsilon)/2^{i} \lambda(G^{(t)},S)$, with probability $1-O(1/n^{4})$. Let $(S^*, V \setminus S^*)$ be a min-cut of $G_{i}^{(t)}$, i.e.,  $\lambda(G_{i}^{(t)}, S^*) = \lambda(G_{i}^{(t)})$. Setting $i= \lfloor \log_2 1/p^{(t)} \rfloor$, we get that:
\[
 \mathbb{E}[\lambda(G^{(t)}_i)] \leq \lambda(G^{(t)})/2^{i}  \leq 2p^{(t)} \lambda(G^{(t)}) \leq 24 \log n/\varepsilon^{2}.
\]
The latter along with the implication of Lemma \ref{lemm: Karger} give that for any $\varepsilon \in (0,1)$, the size of the minimum cut in $G^{(t)}_{i}$ is at most $(1+\varepsilon) 24 \log n / \varepsilon^{2} \leq 48 \log n / \varepsilon^{2}$ with probability $1-O(1/n^{4})$. Thus, $j \leq \lfloor \log_2 1 / p^{(t)} \rfloor$ with probability $1-O(1/n^{4})$. Additionally, we observe that the algorithm returns a $(1+O(\varepsilon)) $-approximation to a min-cut of $G^{(t)}$ w.h.p. since by Lemma \ref{lemm: Karger}, $2^{i} \lambda(G_i^{(t)})/(1-\varepsilon) \leq (1+\varepsilon)/(1-\varepsilon)\lambda(G^{(t)}) = (1+O(\varepsilon))\lambda(G^{(t)})$ w.h.p. Note that for any $t$, $\lfloor \log_2 1 / p^{(t)} \rfloor \leq \lfloor \log n \rfloor$, and thus it is sufficient to maintain only $O(\log n)$ sampled subgraphs.

Since our algorithm applies to unweighted simple graphs, we know that $t \leq O(n^{2})$. Now applying union bound over all $t \in \{1,\ldots O(n^{2})\}$ gives that the probability that the algorithm does not maintain a $1 + O(\varepsilon)$-approximation is at most $O(1/n^2)$.

The total expected time for maintaining a sampled subgraph is $O(m\alpha(n) \log^{2} n / \varepsilon^{2})$ and the required space is $O(n \log n / \varepsilon^{2})$ (Corollary \ref{cor: ExactPolyLog}). Maintaining $O(\log n)$ such subgraphs gives an $O(\alpha(n)\log^{3} n / \varepsilon^{2})$ amortized time per edge insertion and an $O(n \log^2 n / \varepsilon^{2})$ space requirement. The $O(\log n)$ query time follows as in the worst case we scan at most $O(\log n)$ subgraphs, each answering a min-cut query in constant time. 
\end{proof}

\subsection{Improving the space to \texorpdfstring{$O(n \log n / \varepsilon^2)$}{O (n log n)}}
We next show how to bring down the space requirement of the previous algorithm to $O(n \log n / \varepsilon^{2})$ without degrading its running time. The main idea is to keep a single sampled subgraph instead of $O(\log n)$ of them. 

Let $G=(V,E)$ be an unweighted undirected graph and assume each edge is given some random weight $p_e$ chosen uniformly from $[0,1]$. Let $G^{w}$ be the resulting weighted graph. For any $p > 0$, we denote by $G(p)$ the unweighted subgraph of $G$ that consists of all edges that have weight at most $p$. We state the following lemma due to Karger~\cite{kargerPHD}:

\begin{lemma} \label{lemm: Karger2} Let $k = 48 \log n / \varepsilon^{2}$.
Given a connected graph $G$, let $p$ be a value such that $p \geq k/ (4 \lambda(G))$. 
Then with high probability, $\lambda(G(p)) \leq k$ and $\lambda(G(p))/p$ is an $(1+\varepsilon)$-approximation to a min-cut of $G$.
\end{lemma}
\begin{proof}
Since the weight of every edge is uniformly distributed, the probability that an edge has weight at most $p$ is exactly $p$. Thus, $G(p)$ is a graph that contains every edge of $G$ with probability $p$. The claim follows from Lemma~\ref{lemm: Karger}.
%As $p \geq 10 \log n / (\varepsilon^{2} \lambda(G))$, by applying Lemma \ref{lemm: Karger} we get that $\lambda(S(p)) \leq 40 \log n / \varepsilon^{2}$ and $\lambda(S(p))/p$ is an $(1+\varepsilon)$-approximation of the min-cut of $G$ with high probability.
\end{proof}

For any graph $G$ and some appropriate weight $p \geq k/(4\lambda(G))$, the above lemma tells us that the min-cut of $G(p)$ is bounded by $k$ with high probability. 
Thus, instead of considering the graph $G$ along with its random edge weights, we build a collection of $k+1$ minimum edge-disjoint spanning forests (using those edge weights). We note that such a collection is an msfd of order $k+1$  for $G$ with $O(kn)$ edges and by Lemma \ref{correctness1}, it preserves all minimum cuts of $G$ up to size $k$.

Our algorithm uses the following two data structures: 

(1) {\textsc{NI-Sparsifier}$(k)$ data-structure}:  Given a graph $G$, where each edge $e$ is assigned some weight $p_e$ and some parameter $k$, we devise an insertion-only data-structure that maintains a collection of $k+1$  minimum edge-disjoint spanning forests $F_1,\ldots,F_{k+1}$ with respect to the edge weights. Let $F = \bigcup_{i\leq k+1} F_i$. Since we are in the incremental setting, it is known that the problem of maintaining a single minimum spanning forest can be solved in time $O(\log n)$ per insertion using the dynamic tree structure of Sleator and Tarjan~\cite{SleatorT83}. Specifically, we use this data-structure to 
determine for each pair of nodes $(u,v)$ the maximum weight of an edge in the cycle that the edge $(u,v)$ induces in the
minimum spanning forest $F_i$. Let \text{max-weight}$(F_i(u,v))$ denote such a maximum weight. The update operation works as follows:  when a new edge $e = (u,v)$ is inserted into $G$, we first use the dynamic tree data structure to test whether $u$ and $v$ belong to the same tree. If no, we link their two trees with the edge $(u,v)$ and return the pair (TRUE, NULL) to indicate that $e$ was added to $F_i$ and no 
edge was evicted from $F_i$. Otherwise, we 
check whether $p_e > \text{max-weight}(F_i(e))$. If the latter holds, we make no changes in the forest and return 
(FALSE, $e$). Otherwise, we replace one of the maximum edges, say $e'$, on the path between $u$ and $v$ in the tree by $e$ and return (TRUE, $e'$). The boolean value that is returned indicates whether $e$ belongs to $F_i$ or not, the second value that is returned gives an edge that does not (or no longer) belong to $F_i$. Note that each edge insertion requires $O(\log n)$ time. We refer to this insert operation as \textsc{Insert-MSF}$(F_i, e, p_e)$. 

Now, the algorithm that maintains the weighted minimum spanning forests implements the following operations:
\begin{itemize}
\item \textsc{Initialize-NI}$(k)$: Initializes the data structure for $k+1$ empty minimum spanning forests.
\item \textsc{Insert-NI}$(e,p_e)$: Set $i = 1$, $e' = e$, taken = FALSE. \\
\hspace*{2.8cm} \textbf{while} ($(i \le k+1$) and $e' \neq \text{NULL}$) \textbf{do}  \\
%\begin{enumerate}
 \hspace*{3.5cm} Set ($t'$, $e''$) = $\textsc{Insert-MSF}(F_i, e', p_{e'})$.\\
 \hspace*{3.5cm} \textbf{if} ($e' = e$) \textbf{then} set taken = $t'$ \textbf{endif} \\
 \hspace*{3.5cm} Set $e' = e''$ and $i = i + 1$.\\
%\end{enumerate}
\hspace*{2.8cm} \textbf{endwhile}\\ 
\hspace*{2.8cm} \textbf{if} ($e' \ne e$) \textbf{then} \textbf{return} (taken, $e'$) \textbf{else} \textbf{return} (taken, NULL).

%If not the output is either $e$ or $e'$, continue examining the forest $S_2$ and the outputted edge. Repeat the same up to $k%$ forests. If at some point the output is NULL, stop,  and return TRUE, if $e$ was added to one of the $k$ forests, otherwise %return FALSE.
\end{itemize}

The boolean value that is returned indicates whether $e$ belongs to $F$ or not, the second value returns an edge that is removed from $F$, if any. 
%We use this information to build a hierarchy of minimum spanning forests.

Recall that $F = \bigcup_{i \le \kk+1} F_i$.
%We note that the above construction guarantees that $S$ is a DA-msfd of order $\kk$   for the graph $G$. 
We use the abbreviation $\textsc{NI-Sparsifier}(k)$ to refer to this data-structure. Throughout the algorithm we will associate
a weight with each edge in $F$ and use $F^w$ to refer to this weighted version of $F$.

\begin{lemma} \label{lemma: NI} For $k > 0$ and any graph $G$, \textsc{NI-Sparsifier$(k)$} maintains a weighted mfsd of order $\kk+1$ of $G$ under edge insertions. The algorithm uses $O(kn)$ space and the total time for inserting $m$ edges is $O(k m\log n)$.
\end{lemma}
\begin{proof}
We first show that \textsc{NI-Sparsifier$(k)$} maintains a forest decomposition such that (1) the forests are edge-disjoint and (2) each forest is maximal. We proceed by induction on the number $m$ of edge insertions. 

For $m=0$, the forest decomposition is empty. Thus the edge-disjointness and maximality of forests trivially hold.
For $m>0$, consider the $m$-th edge insertion, which inserts an edge $e$. Let $F'$, resp. $F$, denote the union of forests before, resp. after, the insertion of edge $e$. By the inductive assumption, $F'$ satisfies (1) and (2). If $F = F'$, i.e., the edge $e$ was not added to any of the forests when \textsc{Insert-NI}$(e,p_e)$ was called, then $F$ also satisfies (1) and (2). Otherwise $F \neq F'$ and note that by construction, $e$ is appended to exactly one forest. Let $F'_j$, resp. $F_j$, denote such maximal forest before, resp. after, the insertion of $e$. We distinguish two cases. If $e$ links two trees of $F'_j$, then $F_j$ is also a maximal forest and forests of $F$ are edge-disjoint. Thus $F$ satisfies (1) and (2). Otherwise, the addition of $e$ results in  the deletion of another edge $e' \in F'_j$. It follows that $F_j$ is maximal and the current forests are edge-disjoint. Applying a similar argument to the addition of edge $e'$ in the remaining forests, we conclude that $F$ satisfies (1) and $(2)$. 

We next argue about time and space complexity. The dynamic tree data structure can be implemented in $O(n)$ space, where each query regarding $\text{max-weight}(F_i(u,v))$ can be answered in $O(\log n)$ time. Since the algorithm maintains $k+1$ such forests, the space requirement is $O(kn)$. The total running time follows since insertion of an edge can result in at most $k+1$ executions of the \textsc{Insert-MSF}$(F_i,e,p_e)$ procedures, each running in $O(\log n)$ time. 
\end{proof}

(2)  {\textsc{Limited Exact Min-Cut}$(k)$ data-structure}:  We use Algorithm \ref{algo: ExactMinCutUpToK} to implement the following operations for
any  unweighted graph $G$ and parameter $k$,
\begin{itemize}
\item \textsc{Insert-Limited}$(e)$:  Executes the insertion of edge $e$ using Algorithm \ref{algo: ExactMinCutUpToK}.
\item \textsc{Query-Limited}$()$: Returns $\lambda$.
\item \textsc{Initialize-Limited}$(G,k)$: Builds a data structure for $G$ with parameter $k$ by executing Step 1 of Algorithm \ref{algo: ExactMinCutUpToK} and then 
\textsc{Insert-Limited}$(e)$ for each edge $e$ in $G$. 
\end{itemize}
We use the abbreviation \textsc{Lim}$(k)$ to refer to such data-structure. 

Combining the above data-structures leads to Algorithm \ref{algo: SmallSpaceMinCut}.

\begin{algorithm}
\caption{\textsc{$(1+\varepsilon)$-Min-Cut with $O(n \log n / \varepsilon^{2})$ Space}}
\begin{algorithmic}[1]
\State Set $k = 48 \log n / \varepsilon^{2}$.
\Statex Set $p = 12 \log n / \varepsilon^{2}$.
\Statex Let $H$ and $F^w$ be empty graphs. 
\State \textsc{Initialize-Limited}$(H,k)$.
\Statex \WHILE $\textsc{Query-Limited()} < k$ ~\textbf{do}
\Statex \SPACE \textbf{Receive the next operation}.
\Statex \SPACE  \IF it is a query~\textbf{then} \Return $\textsc{Query-Limited()}/\min\{1,p\}$. 
\Statex \SPACE \ELSE it is the insertion of an edge $e$,~\textbf{then}
\Statex \SPACE Sample a random weight from $[0,1]$ for the edge $e$ and denote it by $p_e$.
\Statex \SPACE \IF $p_e \leq p$ ~\textbf{then} \textsc{Insert-Limited}$(e)$~\ENDIF 
\Statex \SPACE Set (taken, $e'$) = \textsc{Insert-NI}$(e, p_e)$. 
\Statex \SPACE \SPACE~ \IF taken~\textbf{then}
\Statex \SPACE \SPACE \SPACE Insert $e$ into $F^w$ with weight $p_e$. 
\Statex \SPACE \SPACE \SPACE \IF ($e' \ne$ NULL) \textbf{then} remove $e'$ from $F^w$.
\Statex \SPACE \SPACE ~~\ENDIF
\Statex \SPACE \ENDIF
\Statex  \ENDWHILE
\State  // \texttt{Rebuild Step}
\Statex Set $p = p/2$.
\Statex  Let $H$ be the unweighted subgraph of $F^w$ consisting of all edges of weight at most $p$.
\Statex \GOTO 2.
%\algstore{testcont}
\end{algorithmic}
\label{algo: SmallSpaceMinCut}
\end{algorithm}
\paragraph*{Correctness and Running Time Analysis. }  
%Let $S$ denote the unweighted version of $S^w$.
Throughout the execution of Algorithm \ref{algo: SmallSpaceMinCut}, $F$ corresponds exactly to the msfd of order $\kk+1$ of $G$ maintained by \textsc{NI-Sparsifier}($k$).   
In the following, let $H$ be the graph that is given as input to \textsc{Lim}($k$). 
Thus, by Corollary \ref{cor: ExactPolyLog}, \textsc{Query-Limited}$()$ returns $\min\{k,\lambda(H)\}$, i.e., it returns
$\lambda(H)$  as long as $\lambda(H) \leq k$. 
We now formally prove the correctness. 
\begin{lemma}\label{lem:1}
Let $\epsilon \le 1$, $k = 48 \log n / \varepsilon^2$ and assume that the algorithm is started on an empty graph. As long as $\lambda(G) < k$, we have $H=G$, $p = k/4$, and \textsc{Query-Limited}$()$ returns $\lambda(G)$.
The first rebuild step is triggered after the first insertion that increases  $\lambda(G)$ to $k$ and at that time, it holds that $\lambda(G) = \lambda(H) = k$.
\end{lemma}
\begin{proof}
The algorithm starts with an empty graph $G$, i.e., initially $\lambda(G)= 0$. Throughout the sequence of edge insertions $\lambda(G)$ never decreases. 
We show by induction on the number $m$ of edge insertions that $H=G$  and $p = k/4$ as long as
$\lambda(G) < k$. 

Note that $k/4 \ge 1$ by our choice of $\epsilon$.
For $m = 0$, the graphs $G$ and $H$ are both empty graphs and  $p$ is set to $k/4$. 
For $m > 0$, consider the $m$-th edge insertion, which inserts an edge $e$. Let $G$ and $H$ denote the corresponding graphs after the insertion of $e$. By the inductive assumption, $p = k/4$
and 
 $G \setminus \{e\} = H \setminus \{e \}$. As $p \ge 1$,  $e$ is added to $H$ and, thus, it follows that $G = H$. Hence, $\lambda(H) = \lambda(G)$. If $\lambda(G) < k$ but $\lambda(G \setminus \{e\}) < k$, no rebuild is performed and $p$ is not changed. If $\lambda(G) = k$, then the last insertion was exactly the
 insertion that increased $\lambda(G)$ from $k-1$ to $k$. As $H = G$ before the rebuild,  \textsc{Query-Limited}$()$ returns $k$, triggering the first execution of the rebuild step.
\end{proof}

We next analyze the case that $\lambda(G) \ge k$. In this case, both $H$ and $p$ are random variables, as they depend on the randomly chosen weights for the edges.
Let $F(p)$ be the  unweighted subgraph of $F^w$ that contains all edges of weight at most $p$.

\begin{lemma}\label{lem:H}
Let $N_h(p)$ be the graph consisting of all edges that were inserted after the last rebuild and have weight at most $p$ and let $F^{\text{old}}(p)$ be
$F(p)$ right after the last rebuild. Then it holds that $H = F^{\text{old}}(p) \cup N_h(p)$.
\end{lemma}
\begin{proof}
Up to the first rebuild, $N_h = G$ and $p \ge 1$. Thus $N_h(p)  = N_h = G$.
Lemma \ref{lem:1} shows that until the first rebuild $H=G$. As $F^{\text{old}}(p) = \emptyset$, it follows that $H = G = N_h(p) \cup F^{\text{old}}(p)$ up to the first rebuild.

Immediately after each rebuild step, $N_h = \emptyset$ and $H$ is set to be $F(p)$, thus the claim holds.
After each subsequent edge insertion that does not trigger a rebuild, the newly inserted edge is added to $N_h(p)$ and to $H$ iff its weight is at most $p$. Thus, 
both $N_h(p)$ and $H$ change in the same way, which implies that $H = F^{\text{old}}(p) \cup N_h(p)$.
\end{proof}

\begin{lemma}\label{DA-msfd}
At the time of a rebuild $F(p)$ is an msfd of order $\kk+1$ of $G(p)$.
\end{lemma}
\begin{proof}
\textsc{NI-sparsifier} maintains a maximal spanning forest decomposition based on  minimum-weight spanning forests  $F_1, \dots F_{\kk+1}$ of $G$ using the weights $p_e$.
%This is a DA-msfd of order $\kk$ of $G$. 
Now consider the hierarchical decomposition $F_1(p), \dots, F_{\kk+1}(p)$ of $G(p)$ induced by taking only the edges of weight at most $p$ of each forest $F_i$.
Note that \textsc{NI-sparsifier} would return exactly the same hierarchy $F_1(p), \dots, F_{k+1}(p)$ if only the edges of $G(p)$ were inserted into \textsc{NI-sparsifier}.
Thus $F_1(p), \dots, F_{\kk+1}(p)$ is an msfd of order $\kk+1$ of $G(p)$.
\end{proof}

In order to show that $\lambda(H)/\min\{1,p\}$ is an $(1+\varepsilon)$-approximation of $\lambda(G)$ with high probability, we need to show that if $\lambda(G) \geq k$ then
(a) the random variable $p$
is at least $k / (4 \lambda(G))$ w.h.p., which implies that $\lambda(G(p))$ is a $(1 + \varepsilon)$-approximation of $\lambda(G)$ w.h.p. and (b) that $\lambda(H) = \lambda(G(p))$ (by Lemma \ref{lemm: Karger2}).
%, at any point in Algorithm \ref{algo: SmallSpaceMinCut}.  
%From the above, it suffices to prove the following lemma:
\begin{lemma} Let $\varepsilon \le 1$. If $\lambda(G) \geq k$, then (1) $p \geq k / (4\lambda(G))$ with probability $1- O(\log n/n^4)$ and
(2) $\lambda(H) = \lambda(G(p))$. 
\end{lemma}
\begin{proof}
For any $i \geq 0$, after the $i$-th rebuild we have $p = p^{(i)} := 12 \log n / (2^{i}\varepsilon^{2})$. Let $\ell = \lfloor \log (12 \log n / \varepsilon^2)\rfloor$ denote the index of the last rebuild at which $p^{(i)} \geq 1$. For any $i \geq \ell + 1$, we will show by induction on $i$ that (1)  $p^{(i)} = 12 \log n / (2^{i}\varepsilon^{2}) \geq 12 \log n / (\varepsilon^{2}\lambda(G))$ with probability $1-O((i-1-\ell)/n^{4})$,
which is equivalent to showing that $\lambda(G) \geq 2^i$ and that (2) at any point between 
the $i-1$-st and the $i$-th rebuild, $\lambda(H) = \lambda(G(p^{(i-1)}))$.

Once we have shown this, we can argue that the number of rebuild steps is small, thus giving the claimed probability in the lemma. Indeed, note that $\lambda(G) \leq n$ since $G$ is unweighted. Additionally, from above we get that after the $i$-th rebuild, $\lambda(G) \geq 2^{i}$ with high probability. Combining these two bounds yields $i \leq O(\log n)$ w.h.p., i.e., the number of rebuild steps is at most $O(\log n)$.

We first analyse $i=\ell+1$. Note that $\ell+1$ is the index of the first rebuild at which $p^{(i)} < 1$. Assume that the insertion of some edge $e$ caused the first rebuild. Lemma~\ref{lem:1} showed that (1) at the first rebuild $\lambda(G) = k$
and (2) that up to the first rebuild $G(p) = G = H$. We observe that (1) and (2) remain true up to the $(\ell+1)$-st rebuild. In addition, $\lambda(G) = k \geq 24 \log n / \varepsilon^2 \geq 2^{i}$, which implies that $p^{(i)} \geq 1/2$. This shows the base case.

For the induction step ($i > \ell+1$), we inductively assume that (1) at the $(i-1)$-st rebuild, $p^{(i-1)} \geq 12 \log n / (\varepsilon^{2} \lambda(G^{\text{old}}))$ with probability $1- O((i-2-\ell)/n^{4})$, where $G^{\text{old}}$ is the graph $G$ right before the insertion that triggered the $i$-th rebuild
(i.e., at the last point in time when \textsc{Query-Limited}$()$ returned a value less than $k$),
and (2) that $\lambda(H) = \lambda(G(p^{(i-2)}))$ at any time between the $(i-2)$-nd and the $(i-1)$-st rebuild.
Let $e$ be the edge whose insertion caused the $i$-th rebuild.
Define $G^{\text{new}} = G^{\text{old}} \cup \{e\}$. By induction hypothesis, with probability $1-O((i-2-\ell)/n^4)$, $p^{(i-1)} \geq 12 \log n / (\varepsilon^{2}\lambda(G^{\text{old}})) \geq 12 \log n / (\varepsilon^{2}\lambda(G^{\text{new}}))$ as $\lambda(G^{\text{old}}) \leq \lambda(G^{\text{new}})$. Thus, by Lemma \ref{lemm: Karger2}, we get that $\lambda(G^{\text{new}}(p^{(i-1)}))/p^{(i-1)} \leq (1+\varepsilon) \lambda(G^{\text{new}})$ with probability $1-O(1/n^4)$. Applying an union bound, we get that the two previous statements hold simultaneously with probability $1-O((i-1-\ell)/n^4)$.

 We show below that $\lambda(G^{\text{new}}(p^{(i-1)})) = \lambda(H^{\text{new}})$, where
$H^{\text{new}}$ is the graph stored in \textsc{Lim}($k$) right before the $i$-th rebuild.
Thus, $\lambda(H^{\text{new}}) = k$, which implies that  
\begin{align*}
	\lambda(G^{\text{new}}(p^{(i-1)})) = k =48 \log n / \varepsilon^{2} & \leq  (1+\varepsilon) \lambda(G^{\text{new}}) \cdot  p^{(i-1)} \\
	 & = (1+\varepsilon) \lambda(G^{\text{new}}) \cdot 12 \log n /(2^{i-1}\varepsilon^{2}),
\end{align*}
with probability $1-O((i-1-\ell)/n^4)$. This in turn implies that with probability $1-O((i-1-\ell)/n^4)$, $\lambda(G^{\text{new}}) \geq 2^{i+1}/(1+\varepsilon) \geq 2^{i}$ by our choice of $\varepsilon$. 

It remains to show that $\lambda(G^{\text{new}}(p^{(i-1)})) = \lambda(H^{\text{new}})$. Note that this
is a special case of (2), which claims that at any point between that $(i-1)$-st and the $i$-th rebuild
$\lambda(H) = \lambda(G(p^{(i-1)}))$, where $H$ and $G$ are the current graphs. Thus, to complete the proof of the lemma it suffices to show (2).

 As $H$ is a subgraph of $G(p^{(i-1)})$, we know that $\lambda(G(p^{(i-1)})) \ge \lambda(H)$.
Thus, we only need to show that $\lambda(G(p^{(i-1)})) \le \lambda(H)$.
Let $G^{i-1}$, resp.~$F^{i-1}$,  resp.~$H^{i-1}$, be the graph $G$, resp.~$F$, resp.~$H$, right after rebuild $i-1$ and let $N_h$ be the set of
 edges inserted since, i.e., $G = G^{(i-1)} \cup N_h$. 
 As we showed in Lemma~\ref{lem:H}, $H = F^{i-1}(p^{(i-1)}) \cup N_h(p^{(i-1)})$. Thus, $H^{i-1} = F^{i-1}(p^{(i-1)})$.
 Additionally, 
 by Lemma~\ref{DA-msfd}, $F^{i-1}(p^{(i-1)})$ is an msfd of order $\kk+1$ of
 $G^{i-1}(p^{(i-1)})$. Thus by Lemma \ref{lemm: Nagamochi}, for every cut $(A, V\setminus A)$ of value at most $\kk$ in $H^{i-1}$, $\lambda(H^{i-1},A)  = \lambda(F^{i-1}(p^{(i-1)}),A) =\lambda(G^{i-1}(p^{(i-1)}),A)$,
 where $\lambda(G, A) = |E_G(A,V \setminus A)|$. 
 Now assume towards contradiction that $\lambda(G(p^{(i-1)})) > \lambda(H)$ and consider a minimum cut $(A, V\setminus A)$ in $H$, i.e., $\lambda(H) = \lambda(H,A)$.
 We know that at any time $k \ge \lambda(H).$ Thus $k \ge \lambda(H) = \lambda(H,A)$, which implies $k \ge \lambda(H^{i-1},A)$. By Lemma \ref{lemm: Nagamochi} it follows that $\lambda(H^{i-1},A) = \lambda(G^{i-1}(p^{(i-1)}),A)$.
 Note that $H = H^{i-1} \cup N_h(p^{(i-1)})$ and $G(p^{(i-1)}) = G^{i-1}(p^{(i-1)}) \cup N_h(p^{(i-1)})$. Let $x$ be the number of edges of  $N_h(p^{(i-1)})$ that cross the cut $(A, V\setminus A)$.
 Then $\lambda(H) = \lambda(H,A) = \lambda(H^{i-1}, A) + x = \lambda(G^{i-1}(p^{(i-1)}),A) + x = \lambda(G(p^{(i-1)}),A)$, which contradicts the assumption that
 $\lambda(G(p^{(i-1)})) > \lambda(H)$.
\end{proof}

Since our algorithm is incremental and applies only to unweighted graphs, we know that there can be at most $O(n^{2})$ edge insertions. The above lemma implies that for any current graph $G$, Algorithm \ref{algo: SmallSpaceMinCut} returns a $(1+\varepsilon)$-approximation to a min-cut of $G$ with probability $1-O(\log n/n^{4})$. Applying an union bound over $O(n^{2})$ possible different graphs, gives that the probability that the algorithm does not maintain a $(1+\varepsilon)$-approximation is at most $O(\log n/n^2) = O(1/n)$. Thus, at any time we return a $(1+\varepsilon)$-approximation with probability $1-O(1/n)$.

\begin{theorem} \label{thm: space2}
There is an $O(n \log n/\varepsilon^{2})$ space randomized algorithm that processes a stream of edge insertions starting from an empty graph $G$ and maintains a $(1+\varepsilon)$-approximation to a min-cut of $G$ with high probability. The 
total time for insertiong $m$ edges is 
$O(m \alpha(n)\log^{3} n / \varepsilon^{2})$ 
%resp. and $O(m \alpha(n)\log^{2} n / \varepsilon^{2} + n \alpha(n)\log^{4} n / \varepsilon^{2})$
and queries can be answered in constant time.
\end{theorem}
\begin{proof}
The space requirement is $O(n \log n/ \varepsilon^{2})$ since at any point of time, the algorithm keeps $H$, $F^{w}$,
\textsc{Lim}($k$), and \textsc{NI-Sparsifier} ($k)$, each of size at most $O(n \log n/ \varepsilon^{2})$ (Corollary \ref{cor: ExactPolyLog} and Lemma \ref{lemma: NI}).

When Algorithm \ref{algo: SmallSpaceMinCut} executes a \texttt{Rebuild Step}, only the \textsc{Lim}($k)$ data-structure is rebuilt, but not \textsc{NI-Sparsifier}($k$). During the whole algorithm $m$ \textsc{Insert-NI} operations are
performed. Thus, by Lemma \ref{lemma: NI},  the total time for all operations involving $\textsc{NI-Sparsifier}(k)$ is $O(m\log^2 n / \varepsilon^{2})$.
% and the space usage is $O(n \log n / \varepsilon^{2})$. 

%We define $m_H$ to be the number of edges in $H$.
It remains to analyze Steps $2$ and $3$. % and the operations involving the data-structure $\textsc{Lim}(k)$. 
In Step 2, \textsc{Initialize-Limited}$(H,k)$ takes at most $O(m \alpha(n)\log^{2} n / \varepsilon^{2})$ total time (Corollary \ref{cor: ExactPolyLog}). 
The running time of Step $3$ is $O(m)$ as well. Since the number of \texttt{Rebuild Steps} is at most $O(\log n)$,
%MH: WE SHOULD ACTUALLY PROVE THIS
it follows that the total time for all
\textsc{Initialize-Limited}$(H,k)$ calls in Steps $2$ and the total time of Step $3$ throughout the execution of the algorithm is $O(m \alpha(n)\log^{3} n / \varepsilon^{2})$.
%$= O(n \alpha(n) \log^3n/ \varepsilon^2)$. 

We are left with analyzing the remaining part of Step 2. Each query operation executes one \textsc{Query-Limited}() operation, which takes constant time.
Each insertion executes one \textsc{Insert-NI}($e,p_e$) operation, which takes amortized time
%and potentially one \textsc{Insert-Limited}($e)$ operation. The earlier operation takes amortized time
$O(\log^2n/ \varepsilon)$. We maintain the edges of $F^w$ in a balanced binary tree so that each insertion and deletion takes $O(\log n)$ time.
As there are $m$ edge insertions the  remaining part of Step 2 takes total time
$O(m \log^{2} n / \varepsilon^{2})$. Combining the above bounds gives the theorem. \end{proof}

\paragraph*{Acknowledgements}

The research leading to these results has received funding from the European Research Council under the European Union's 7th Framework Programme (FP/2007-2013) / ERC Grant Agreement no.~340506 for M. Henzinger. M. Thorup's research is partly supported by Advanced Grant DFF-0602-02499B from the Danish Council for Independent Research under the Sapere Aude research career programme. This work was done in part while M. Henzinger and M. Thorup were visiting the Simons Institute for the Theory of Computing.

\bibliographystyle{plainurl}
\bibliography{literature}

\begin{thebibliography}{10}

\bibitem{abboud}
Amir Abboud and Virginia~Vassilevska Williams.
\newblock Popular conjectures imply strong lower bounds for dynamic problems.
\newblock In {\em Proc. of the 55th FOCS}, pages 434--443. IEEE, 2014.

\bibitem{AhnG09}
Kook~Jin Ahn and Sudipto Guha.
\newblock Graph sparsification in the semi-streaming model.
\newblock In {\em Proc. of the 36th ICALP}, pages 328--338, 2009.

\bibitem{AhnGM12}
Kook~Jin Ahn, Sudipto Guha, and Andrew McGregor.
\newblock Graph sketches: sparsification, spanners, and subgraphs.
\newblock In {\em Proc. of the 32nd PODS}, pages 5--14, 2012.

\bibitem{BenczurK15}
Andr{\'{a}}s~A. Bencz{\'{u}}r and David~R. Karger.
\newblock Randomized approximation schemes for cuts and flows in capacitated
  graphs.
\newblock {\em {SIAM} J. Comput.}, 44(2):290--319, 2015.

\bibitem{sayan}
Sayan Bhattacharya, Monika Henzinger, Danupon Nanongkai, and Charalampos~E.
  Tsourakakis.
\newblock Space- and time-efficient algorithm for maintaining dense subgraphs
  on one-pass dynamic streams.
\newblock In {\em Proc. of the 47th STOC}, pages 173--182, 2015.

\bibitem{Cormen}
Thomas~H. Cormen, Charles~E. Leiserson, Ronald~L. Rivest, and Clifford Stein.
\newblock {\em Introduction to Algorithms {(3.} ed.)}.
\newblock {MIT} Press, 2009.

\bibitem{dinitz}
E.~A. Dinitz, A.~V. Karzanov, and M.~V. Lomonosov.
\newblock On the structure of a family of minimum weighted cuts in a graph.
\newblock {\em Studies in Discrete Optimization}, pages 290--306, 1976.

\bibitem{DinitzW98}
Yefim Dinitz and Jeffery Westbrook.
\newblock Maintaining the classes of 4-edge-connectivity in a graph on-line.
\newblock {\em Algorithmica}, 20(3):242--276, 1998.

\bibitem{Gabow91}
Harold~N. Gabow.
\newblock Applications of a poset representation to edge connectivity and graph
  rigidity.
\newblock In {\em Proc. of the 32nd FOCS}, pages 812--821, 1991.

\bibitem{Gabow95}
Harold~N. Gabow.
\newblock A matroid approach to finding edge connectivity and packing
  arborescences.
\newblock {\em J. Comput. Syst. Sci.}, 50(2):259--273, 1995.

\bibitem{GalilI93}
Zvi Galil and Giuseppe~F. Italiano.
\newblock Maintaining the 3-edge-connected components of a graph on-line.
\newblock {\em {SIAM} J. Comput.}, 22(1):11--28, 1993.

\bibitem{GibbKKT15}
David Gibb, Bruce~M. Kapron, Valerie King, and Nolan Thorn.
\newblock Dynamic graph connectivity with improved worst case update time and
  sublinear space.
\newblock {\em CoRR}, abs/1509.06464, 2015.

\bibitem{henzinger15}
Monika Henzinger, Sebastian Krinninger, Danupon Nanongkai, and Thatchaphol
  Saranurak.
\newblock Unifying and strengthening hardness for dynamic problems via the
  online matrix-vector multiplication conjecture.
\newblock In {\em Proc. of the 47th STOC}, pages 21--30, 2015.

\bibitem{henzinger97}
Monika~Rauch Henzinger.
\newblock A static 2-approximation algorithm for vertex connectivity and
  incremental approximation algorithms for edge and vertex connectivity.
\newblock {\em Journal of Algorithms}, 24(1):194--220, 1997.

\bibitem{kargerPHD}
David Karger.
\newblock {\em Random Sampling in Graph Optimization Problems}.
\newblock PhD thesis, Stanford University, Stanford, 1994.

\bibitem{Karger94}
David~R. Karger.
\newblock Using randomized sparsification to approximate minimum cuts.
\newblock In {\em Proceedings of the Fifth Annual {ACM-SIAM} Symposium on
  Discrete Algorithms. 23-25 January 1994, Arlington, Virginia.}, pages
  424--432, 1994.

\bibitem{Karger99}
David~R. Karger.
\newblock Random sampling in cut, flow, and network design problems.
\newblock {\em Math. Oper. Res.}, 24(2):383--413, 1999.

\bibitem{kargermin}
David~R. Karger.
\newblock Minimum cuts in near-linear time.
\newblock {\em J. {ACM}}, 47(1):46--76, 2000.

\bibitem{thorup}
Ken{-}ichi Kawarabayashi and Mikkel Thorup.
\newblock Deterministic global minimum cut of a simple graph in near-linear
  time.
\newblock In {\em Proc. of the 47th STOC}, pages 665--674, 2015.

\bibitem{KelnerL13}
Jonathan~A. Kelner and Alex Levin.
\newblock Spectral sparsification in the semi-streaming setting.
\newblock {\em Theory Comput. Syst.}, 53(2):243--262, 2013.

\bibitem{LackiS11}
Jakub Lacki and Piotr Sankowski.
\newblock Min-cuts and shortest cycles in planar graphs in {$O(n \log \log n)$}
  time.
\newblock In {\em Proc. of the 19th ESA}, pages 155--166, 2011.

\bibitem{menger}
Karl Menger.
\newblock Zur allgemeinen kurventheorie.
\newblock {\em Fundamenta Mathematicae}, 1(10):96--115, 1927.

\bibitem{NagamochiI92}
Hiroshi Nagamochi and Toshihide Ibaraki.
\newblock A linear-time algorithm for finding a sparse k-connected spanning
  subgraph of a k-connected graph.
\newblock {\em Algorithmica}, 7(5{\&}6):583--596, 1992.

\bibitem{NagamochiI}
Hiroshi Nagamochi and Toshihide Ibaraki.
\newblock {\em Algorithmic Aspects of Graph Connectivity}.
\newblock Cambridge University Press, New York, NY, USA, 1 edition, 2008.

\bibitem{DS16}
Danupon Nanongkai and Thatchaphol Saranurak.
\newblock Dynamic cut oracle.
\newblock under submission, 2016.

\bibitem{Poutre00}
Johannes A.~La Poutr{\'{e}}.
\newblock Maintenance of 2- and 3-edge-connected components of graphs {II}.
\newblock {\em {SIAM} J. Comput.}, 29(5):1521--1549, 2000.

\bibitem{SleatorT83}
Daniel~Dominic Sleator and Robert~Endre Tarjan.
\newblock A data structure for dynamic trees.
\newblock {\em J. Comput. Syst. Sci.}, 26(3):362--391, 1983.

\bibitem{fullythorup}
Mikkel Thorup.
\newblock Fully-dynamic min-cut.
\newblock {\em Combinatorica}, 27(1):91--127, 2007.

\bibitem{thorupkarger}
Mikkel Thorup and David~R Karger.
\newblock Dynamic graph algorithms with applications.
\newblock In {\em Algorithm Theory-SWAT 2000}, pages 1--9. Springer, 2000.

\end{thebibliography}

\appendix

\section{Missing proofs in Section \ref{sec: sparseCertificates}} 
Next we show a proof for Lemma \ref{lemm: Nagamochi}. The arguments closely follow the work of Nagamochi and Ibaraki~\cite{NagamochiI}. We first present the following helpful lemma. 

\begin{lemma} \label{lemm: sparseHelpful}
Let $\mathcal{F} = (F_1,\ldots,F_m)$ be an \emph{msfd} of order $m$ of a graph $G=(V,E)$. Then for any edge $(u,v) \in F_j$
and any $i \le j$, it holds that $\lambda(u,v,\bigcup_{l \leq i}F_l) \geq i$.
\end{lemma}
\begin{proof}
Fix some edge $e=(u,v) \in F_j$. We first argue that for each $i=1,\ldots,j-1$, the forest $(V,F_i)$ contains some $(u,v)$-path. Indeed, by the maximality of the forest $(V,F_i)$, the graph $(V,F_i \cup \{e\})$ must have some cycle $C$ that contains $e$. Thus, $P=C \setminus e$ is the $(u,v)$-path in the forest $(V,F_i)$. It follows that $(V,\bigcup_{l \leq i}F_l)$ has $i$ edge-disjoint paths. Next, observe that $G_j=(V,\bigcup_{l \leq j}F_l)$ has $j$ edge-disjoint paths, namely the $j-1$ edge disjoint paths
in $G_{j-1}$ (which does not contain the edge $(u,v)$) and
the 1-edge path consisting of  the edge $(u,v)$. Hence, $\lambda(u,v,\bigcup_{l \leq i}F_l) \geq i$,
\end{proof}

\begin{pfof}{Lemma \ref{lemm: Nagamochi}} Assume that $\lambda(S,G) \leq k-1$. Then by definition of $G_k$, we know that $G_k$ preserves any cut $S$ of size up to $k$. Thus $\lambda(S,G_k) = \lambda(S,G)$. 

For the other case, $\lambda(S,G) \geq k$ and assume that $\lambda(S,G_k) < \lambda(S,G)$ (otherwise the lemma follows). Then there is an edge $e=(u,v) \in E_G(S,V \setminus S) \setminus E_{G_k}(S, V \setminus S)$. Since $e \not \in \bigcup_{i \leq k} F_i$, this means that $e$ belongs to some forest $F_j$ with $j > k$. By Lemma \ref{lemm: sparseHelpful}, we have that $\lambda(u,v,G_k) \geq k$. Since $(S,V \setminus S)$ separates $u$ and $v$ in $G_k$, it follows that $\lambda(S,G_k) = |E_{G_k}(S, V \setminus S)| \geq \lambda (u,v,G_k) \geq k$. 
\end{pfof}

\end{document}